\documentclass[american,aps,pra,reprint, superscriptaddress]{revtex4-1}
\usepackage{tikz}
\usepackage{amsmath,amscd,amsthm,amssymb}
\usepackage{mathrsfs}
\usepackage{graphicx}
\usepackage{epsf,epstopdf}
\usepackage[]{caption3}
\usepackage[all]{xy}

\usepackage[unicode=true,pdfusetitle, bookmarks=true,bookmarksnumbered=false,bookmarksopen=false, breaklinks=false,pdfborder={0 0 0},backref=false,colorlinks=false]{hyperref}
\hypersetup{colorlinks,linkcolor=myurlcolor,citecolor=myurlcolor,urlcolor=myurlcolor}
\usepackage{graphics,epstopdf,graphicx, amsthm, amsmath, amssymb, times, braket, color, bm}
\usepackage[up]{subfigure}
\usepackage{cleveref}

\usepackage{diagbox}

\definecolor{myurlcolor}{rgb}{0,0,0.7}

\theoremstyle{plain}
\newtheorem{thm}{\protect\theoremname}
\newtheorem{prop}{Proposition}

\providecommand{\theoremname}{Theorem}
\newcommand*{\myproofname}{Proof}

\makeatother

\theoremstyle{definition}

\newtheorem{exmp}{Example}

\theoremstyle{remark}
\newtheorem{rem}{Remark}

\begin{document}

 \author{Chunhe Xiong}
 \email{xiongchunhe@csu.edu.cn}

\affiliation{School of Mathematics and Statistics, HNP-LAMA, Central South University, Changsha 410083, China}

 \author{Sunho Kim}
 \email{kimsunho81@hrbeu.edu.cn}
 \affiliation{School of Mathematical Sciences, Harbin Engineering University, Harbin 150001, China}

 \author{Long Long}

\affiliation{School of Mathematics and Statistics, HNP-LAMA, Central South University, Changsha 410083, China}

\author{Junde Wu}
\email{wjd@zju.edu.cn}
\affiliation{School of Mathematical Sciences, Zhejiang University, Hangzhou 310027, China}

\title{Private Capacity of Quantum Channels Induced by Non-stabilizer Environmental States}

\begin{abstract}

We investigate the private capacity of quantum channels using the recently proposed quantum convolution theory for discrete-variable quantum systems. We focus on the role of the magic resource played in this framework. %
Firstly, for a large class of convolutional channels, we find that the private capacity is zero if the fixed environmental state is a stabilizer state. Moreover, we show that the private capacity can be nonzero for some magic environmental states.
Furthermore, we show that the private capacity of a discrete beam splitter unitary is upper-bounded by the amount of magic of the environmental state. In addition, if the environmental state exhibits a certain symmetric structure, even if it is magic, the corresponding private capacity will also vanish for a class of convolution. These results emphasize the role of magic resources in quantum communication.

\end{abstract}


\maketitle

\section{Introduction}

The stabilizer formalism, rooted in the algebraic structure of Pauli operators and Clifford gates, has long served as a cornerstone for quantum error correction and fault-tolerant computation. The common eigenvectors of an Abelian subgroup of the Pauli group are called pure stabilizer states \cite{gottesman1997phd}.
Fundamental classes of quantum error-correcting codes, including Shor codes and toric codes \cite{shor1995code,kitaev2003code}, belong to the category of stabilizer codes, which utilize stabilizer states as carriers.
Especially, the surface code has emerged as a particularly promising candidate for realizing scalable quantum computation \cite{fowler2012,google2023}.

A striking result is that quantum circuits comprise Clifford unitaries acting on stabilizer inputs and measurements, which allows efficient classical simulation, which is also known as the Gottesman-Knill theorem \cite{gottesman1998gkt,Nielsen10}. This theorem underscores the insufficiency of stabilizer resources for achieving quantum computational advantage, yet it simultaneously ignited the pursuit of non-stabilizer resources, also known as " magic, " as essential catalysts for quantum speedups. Such a dichotomy positions non-stabilizer states at the critical frontier between classical simulability and genuine quantum complexity. Recent advances in resource theory have rigorously quantified nonstabilizer states as indispensable for quantum advantage \cite{Veitch2014}, and many measures are proposed to quantify the magic resource \cite{Veitch2012, Bravyi2016,Bravyi2016prx,howard2017,Beverland2020qst,Seddon2021prxq,Leone2022prl,Liu2022prxq,Bu2022cmp,Bu2019prl,Bu2023qst,Bu2024cmp,Chen2024prb,Bu2025tit}.

Quantum information theory establishes fundamental limits on communication through various channel capacities \cite{Bennett2004sci,Holevo2020qe}. The capacities of a quantum channel include the classical \cite{Holevo1973ppi,schumacher1997pra}, private \cite{bennett1984,Cai2004pit,Devetak2005tit}, and quantum capacity \cite{Devetak2005tit,Lloyd1997,Barnum1998pra,Barnum2000tit}, depending on the information to be transmitted. The quantum channel capacity exhibits superadditivity behavior \cite{Divincenzo1998,Horodecki2005,Graeme2007prl,Graeme2008sci,Graeme2009prl,Graeme2009prl2,Hastings2009,Graem2011np,Brandao2012prl,Elkouss2015prl,Leung2014prl,Cubitt2015nc,Lim2018pra,Lim2019pra,Leditzky2018prl,leditzky2023prl,Bausch2021siam,Siddhu2021nc,Siddhu2021tit,Sidhardh2022pra,Filippov2021jpa,Koudia2022,wu2025pra}, which makes it richer and more complex than the classical theory \cite{Shannon1948,Cover1999}. The classical capacity of a quantum channel quantifies how much classical information it can transmit asymptotically with vanishing error \cite{Holevo1998tit,schumacher1997pra}.  The additivity of classical capacity has been rigorously established for various classes of quantum channels \cite{King2002JMP,King2003,Shor2002cmp}, implying that classical capacity completely characterizes their ability to reliably transmit classical information. However, Hastings constructed a random channel to show that Holevo information is superadditive \cite{Hastings2009}, but explicit examples are still lacking \cite{Holevo2020qe,leditzky2023prl}.

The situation for quantum capacity, which characterizes the quantum information that a quantum channel can reliably transmit, is quite different. There are many concrete channels with superadditive coherent information \cite{Divincenzo1998,Graeme2007prl,Leditzky2018prl, Bausch2021siam,Siddhu2021nc,Siddhu2021tit,Sidhardh2022pra,Filippov2021jpa,leditzky2023prl}.  Remarkably, Smith and Yard show theoretically that two quantum channels, each with zero quantum capacity, can have a nonzero capacity when used together, which is called the superactivation phenomenon \cite{Graeme2008sci}. Furthermore, Smith, Smolin, and Yard show that this superactivation phenomenon also occurs in a continuous variable quantum system \cite{Graem2011np}.

Private capacity characterizes the ability of a quantum channel to reliably transmit private information, being intercepted by eavesdroppers. The fact that quantum information transmission will operationally derive private information transmission implies that private capacity is a natural upper bound of quantum capacity. In \cite{Horodecki2008tit}, Horodecki et al. show that there exist low-dimensional bound entangled states with a one-way distillable cryptographic key, which implies that there exists a quantum channel with zero quantum capacity but non-zero private capacity. This result not only establishes a complete separation between the quantum capacity and the private capacity of a quantum channel but also provides a foundational basis for the superactivation phenomenon of quantum capacity \cite{Graeme2008sci}.

We focus on the private capacity of quantum convolutional channels in a discrete variable quantum system. Recently, Bu, Gu, and Jaffe defined a quantum convolution and developed a framework to study discrete variable quantum systems. With this, they discovered a quantum central limit theorem, which means that repeated quantum convolution with a given state converges to a stabilizer state \cite{Bu2023pnas,Bu2025tit}. Moreover, they develop a theoretical method to test stabilizers and to measure magic resource \cite{bu2025cmp}. Furthermore, they bridge quantum communication and magic resource theory in the discrete variable quantum system by investigating the channels derived from the discrete beam splitter unitary with the fixed environmental state \cite{Bu2025prl}. Actually, they show that a magic resource is necessary to obtain the nonzero quantum capacity by proving that the corresponding quantum capacity is zero if the environmental state is stabilizer or symmetric magic, while the quantum capacity can be nonzero for some magic states.

In this paper, we investigate the private capacity of a kind of discrete variable global unitary with fixed environmental states. The main results are as follows: (1) If the fixed environmental state is the mixture of pure stabilizer states, the private capacity of corresponding convolutional channels is zero for these unitary correspond to positive and invertible parameter matrices, including discrete beam splitter and discrete amplifier. (2) We find that the private capacity of a discrete beam splitter channel is upper-bounded by the magic resource of a pure environmental state. This reinforces the result presented in \cite{Bu2025prl}, demonstrating that even private capacity is an upper bound of quantum capacity, it remains governed by the magic resources of the environmental state for the discrete beam splitter, consistent with the latter. (3) The private capacity of the above two unitary evolutions can be nonzero for some magic states. Since the symmetric magic environmental state derives an anti-degradable channel for discrete beam splitter unitary \cite{Bu2025prl}, it means that the magic resource is also a necessary condition to obtain private capacity for this channel. These results are similar to quantum capacity.




 \section{Basic framework}

\subsection{Stabilizer formalism}
 Let us consider n-qudit system with Hilbert space $\mathcal{H}^{\otimes n}$, where $\mathcal{H}\simeq\mathcal{C}^d$. Let $L(\mathcal{H}^{\otimes n})$ denote the set of all linear operators on $\mathcal{H}^{\otimes n}$, and $D(\mathcal{H}^{\otimes n})$ denote the set of all quantum states on $\mathcal{H}^{\otimes n}$, namely positive operators with unit trace. If the quantum state $\rho$ is rank-one, then it is called a pure state, also denoted $\rho=\ket{\psi}\bra{\psi}$. In the Hilbert space $\mathcal{H}$, we define one orthonormal set to be a computational basis and denote it in Dirac notation$
\{\ket{k}\}_{k\in\mathbb{Z}_d},
$
 with $\mathbb{Z}_d$ the cyclic group of order $d$. The Pauli matrices X and Z are unitary transformations that act as
\begin{align}
X\ket{k}=\ket{k+1}, Z\ket{k}=\omega^k_d\ket{k},\forall k\in\mathbb{Z}_d,
\end{align}
where $\omega_d=\exp(2\pi i/d)$ is a $d$-th root of unity. If the local dimension d is an odd prime number, the Weyl operators are defined as
$
w(p,q)=\omega^{-2^{-1}pq}_dZ^pX^q.
$
Here $2^{-1}$ denotes the inverse $\frac{d+1}{2}$ of 2 in $\mathbb{Z}_d$. In the $n$-qudit quantum system, the Weyl operator $w(\vec{p},\vec{q})$ is defined as
$
w(\vec{p},\vec{q})=w(p_1,q_1)\otimes...\otimes w(p_n,q_n),
$
with $\vec{p}=(p_1,...,p_n),\vec{q}=(q_1,...,q_n)\in\mathbb{Z}^n_d$.

Denote
$
V^n=\mathbb{Z}^n_d\times\mathbb{Z}^n_d
$
as the phase space of the $n$ -qudit system, the Weyl operators $\{w(\vec{p},\vec{q})\}$ in the phase space form an orthonormal basis in $L(\mathcal{H}^{\otimes n})$ with respect to the inner product $\langle A,B\rangle=\frac{1}{d^n}\textmd{tr}[A^{\dagger}B]$. For each $\vec{z}=(\vec{p},\vec{q}),\vec{z}^{\prime}=(\vec{p}^{\prime},\vec{q}^{\prime})\in V^n$, define symplectic inner product as
\begin{align}
\langle\vec{z},\vec{z}^{\prime}\rangle_s:=\vec{p}\cdot\vec{q}^{\prime}-\vec{p}^{\prime}\cdot\vec{q}\in\mathbb{Z}_d,
\end{align}
where $\vec{p}\cdot\vec{q}=\sum_jp_jq_j$ for $\vec{p}=(p_1,...,p_n), \vec{q}=(q_1,...,q_n)\in \mathbb{Z}^n_d$. With this symplectic inner product, the multiplication between Weyl operators can be written as
\begin{align*}
&w(\vec{z})w(\vec{z}^{\prime})=\omega_d^{2^{-1}\langle\vec{z},\vec{z}^{\prime}\rangle_s}w(\vec{z}+\vec{z}^{\prime}),\\
&w(\vec{z}^{\prime})w(\vec{z})w^{\dagger}(\vec{z}^{\prime})=\omega_d^{\langle\vec{z}^{\prime},\vec{z}\rangle_s}w(\vec{z}).
\end{align*}

For any n-qudit state $\rho\in\mathcal{D}(\mathcal{H}^{\otimes n})$, the characteristic function $\Xi_{\rho}:V^n\rightarrow \mathbb{C}$ is defined as
\begin{align}\label{defn:char-func}
\Xi_{\rho}(\vec{p},\vec{q})=\textmd{tr}[\rho w(-\vec{p},-\vec{q})].
\end{align}

Hence, the state $\rho$ can be written as a linear combination of Weyl operators with characteristic function
$
\rho=\frac{1}{d^n}\sum_{(\vec{p},\vec{q})\in V^n}\Xi_{\rho}(\vec{p},\vec{q})w(\vec{p},\vec{q}).
$ Since $\rho$ is self-adjoint, then it holds that $\Xi_{\rho}(\vec{p},\vec{q})=\Xi^{\ast}_{\rho}(-\vec{p},-\vec{q})$.


A pure stabilizer state $\ket{\psi}$ for an $n$-qudit system is a common eigenstate of an abelian subgroup with $n$ generators $\{g_i\}$, satisfying $g_i\ket{\psi}=\ket{\psi}$ for each $i$. 
A mixed stabilizer state is a convex combination of pure stabilizer states, and the set of stabilizer states is denoted by {\em STAB}.

A quantum state $\rho$ is a minimal stabilizer projection state associated with an abelian subgroup generated by $\{w(p_i, q_i)\}_{i\in[r]}$, if it has the following form \cite{Bu2023discrete,Bu2023pnas}
\begin{align*}
\rho = \frac{1}{d^{n-r}} \prod_{j=1}^r \mathbb{E}_{k_j \in \mathbb{Z}_d} \left[ e^{\frac{2\pi x_ji}{d}}  w(p_j, q_j) \right]^{k_j},
\end{align*}
for some $(x_1,\ldots,x_r) \in \mathbb{Z}_d^r$, with $\mathbb{E}_{k_i \in \mathbb{Z}_d}(\cdot):=\frac{1}{d} \sum_{k_i \in \mathbb{Z}_d} (\cdot)$. Denote this set of states by $\mathrm{MSPS}$, then each pure stabilizer state is an $\mathrm{MSPS}$, while a mixed state of $\mathrm{MSPS}$ can also be regarded as a stabilizer since it preserves the stabilizer formalism. Therefore, the relative entropy of magic can be defined as
\begin{align*}
\textmd{MRM}(\rho):=\min_{\sigma\in\mathrm{MSPS}}D(\rho||\sigma),
\end{align*}
where $D(\rho||\sigma):=\textmd{tr}[\rho(\log\rho-\log\sigma)]$ is the quantum relative entropy. Moreover, it has also proved that $\textmd{MRM}(\rho)=S(\mathcal{M}(\rho))-S(\rho)$ with $\mathcal{M}(\rho)$ is the mean state with characteristic function $\Xi_{\mathcal{M}(\rho)}(\vec{p},\vec{q})=\Xi_{\rho} (\vec{p},\vec{q})$ for $|\Xi_{\rho} (\vec{p},\vec{q})|=1$, otherwise $\Xi_{\mathcal{M}(\rho)}(\vec{p},\vec{q})=0$ \cite{Bu2023pnas,Bu2023discrete}.


\subsection{Channel capacities}
A quantum channel $\Lambda$ is mathematically described by an isometric map $V:\mathcal{H}_A\rightarrow\mathcal{B}\otimes\mathcal{H}_E$. Then the channel and its complement are
$\Lambda(\rho)=\textmd{tr}_EV\rho V^{\dagger}$ and $\Lambda^c(\rho)=\textmd{tr}_BV\rho V^{\dagger}.$
Denote $S(\rho)=-\textmd{tr}\rho\log\rho$ as the von Neumann entropy
(log is always the binary logarithm). Then the Holevo information, coherent information, and private information are obtained from the optimization of entropy expressions as
\begin{align*}
&\chi_{\{p_i,\rho_i\}}(\Lambda)=S\big(\Lambda(\sum_ip_i\rho_i)\big)-\sum_ip_iS(\Lambda(\rho_i)),\\
&P^{(1)}(\Lambda)=\max_{\{p_i,\rho_i\}}\big(\chi_{\{p_i,\rho_i\}}(\Lambda)-\chi_{\{p_i,\rho_i\}}(\Lambda^c)\big),\\
&Q^{(1)}(\Lambda)=\max_{\rho}\big(S(\Lambda(\rho))-S(\Lambda^{c}(\rho))\big).
\end{align*}
The Devetak \cite{Devetak2005tit} and LSD Theorem \cite{Lloyd1997,Shor2002lecture,Devetak2005tit} gives the quantum and private capacity of $\Lambda$ as the regularization of single-letter form:
\begin{align*}
&P^{(1)}(\Lambda)\le P(\Lambda):=\lim_{n\rightarrow\infty}\frac{1}{n}P^{(1)}(\Lambda^{\otimes n}),\\
&Q^{(1)}(\Lambda)\le Q(\Lambda):=\lim_{n\rightarrow\infty}\frac{1}{n}Q^{(1)}(\Lambda^{\otimes n}).\\
\end{align*}



It has been proved that the quantum capacity is a lower bound of the private capacity, that is, $Q\le P$, and the equality holds for the degradable channels, i.e., there exists another quantum channel $\mathcal{N}$ such that $\Lambda^c=\mathcal{N}\circ\Lambda$. Actually, these two kinds of quantum capacities are equal to zero for anti-degradable channels, i.e., there exists another quantum channel $\mathcal{N}$ such that $\Lambda=\mathcal{N}\circ\Lambda^c$.



\subsection{Quantum convolution theory}

In this paper, we consider the quantum convolutional channel derived from a quantum convolution unitary (introduced by Bu et al \cite{Bu2023pnas}) and fixed environmental states. Let us review the framework of quantum convolution theory as follows.

Given a prime number $d$ and define the $2\times2$ invertible matrix of parameters,
\begin{align}\label{parameter-matrix}
G=\begin{bmatrix}
g_{00}&g_{01}\\g_{10}&g_{11}
\end{bmatrix},
\end{align}
with entries in $\mathbb{Z}_d$, satisfying $\det G=g_{00}g_{11}-g_{01}g_{10}\ne0,\mod d$. We  call $G$ nontrivial if at most one of $g_{ij}$s is $0 \mod d$, and call $G$ positive if each $g_{ij}\ne0 \mod d$. With nontrivial $G$, the convolutional unitary $U$ on $\mathcal{H}^{\otimes n}_A\otimes\mathcal{H}^{\otimes n}_B$ is defined as \cite{Bu2023pnas,Bu2023discrete}
\begin{align}\label{key-unitary}
U=\sum_{\vec{i},\vec{j}}\ket{Ng_{11}\vec{i}-Ng_{10}\vec{j},-Ng_{01}\vec{i}+Ng_{00}\vec{j}}\bra{\vec{i},\vec{j}},
\end{align}
where the vectors $\ket{\vec{i}}=\ket{i_1}\otimes...\otimes\ket{i_n}\in\mathcal{H}^{\otimes n}_A$ and $\ket{\vec{j}}=\ket{j_1}\otimes...\otimes\ket{j_n}\in\mathcal{H}^{\otimes n}_B$.

Then, for $\rho\in \mathcal{H}^{\otimes n}_A$ and $\sigma\in\mathcal{H}^{\otimes n}_B$, the quantum convolutional channels $\Lambda$ based on the unitary $U$  defined in (\ref{key-unitary}) is
\begin{align*}
\Lambda(\rho_{AB})=\textmd{tr}_B[U(\rho_{AB})U^{\dagger}].
\end{align*}
Therefore, the convolution of $\rho$ and $\sigma$ is defined as
\begin{align*}
\rho\boxtimes\sigma:=\Lambda(\rho\otimes\sigma).
\end{align*}
By emphasizing the environmental state $\sigma$, the quantum convolutional channel on $\mathcal{H}^{\otimes n}_A$ can be also defined as
\begin{align*}
\Lambda_{\sigma}(\rho)=\textmd{tr}_B[U(\rho\otimes\sigma)U^{\dagger}].
\end{align*}

Similarly, one can define the quantum convolutional by tracing subsystem A as
\begin{align*}
\rho\tilde{\boxtimes}\sigma:=\textmd{tr}_A(U\rho\otimes\sigma U^{\dagger}).
\end{align*}
Then the corresponding quantum convolutional channel is defined as
\begin{align*}
 &\tilde{\Lambda}_{\sigma}(\rho):=\rho\tilde{\boxtimes}\sigma=\textmd{tr}_A[U(\rho\otimes\sigma)U^{\dagger}].
\end{align*}
Noting that if $\sigma$ is a pure state, $\tilde{\Lambda}_{\sigma}$ is also the complementary channel of $\Lambda_{\sigma}$, that is, $\tilde{\Lambda}_{\sigma}=\Lambda^c_{\sigma}$.

Recall the Convolution-multiplication duality as follows (See Proposition 42 in \cite{Bu2023discrete}).

\begin{prop}\label{CMD}
Given two $n$-qudit states $\rho$ and $\sigma$, the characteristic function satisfies
\begin{align*}
 \Xi_{\rho\boxtimes\sigma}(\vec{p},\vec{q})=\Xi_{\rho}(Ng_{11}\vec{p},g_{00}\vec{q})\Xi_{\sigma}(-Ng_{10}\vec{p},g_{01}\vec{q}),
\end{align*}
$\forall \vec{p},\vec{q}\in\mathbb{Z}_d^n$.
\end{prop}





\section{Main results}


In this part, we first consider the private capacity of quantum convolutional channels.

\begin{thm}\label{stabilizer-derives-zero-capacity}
For an invertible matrix of parameters $G$ with $g_{01}g_{10}\ne0$ and the key unitary defined in (\ref{key-unitary}), the private capacity of the corresponding quantum convolutional channels is zero for the mixture of stabilizer pure states, that is,
\begin{align*}
P(\Lambda_{\sigma})=Q(\Lambda_{\sigma})=0,
\end{align*}
for $\sigma=\sum_ip_i\tau_i$ with $\tau_i$s are all pure stabilizer states.
\end{thm}

Actually, this result is similar to the quantum capacity case for quantum discrete beam splitter \cite{Bu2025prl}, where Bu and Jaffe showed that if the environmental state is a pure and stabilizer state, then the corresponding channel is a measure-and-prepare channel.


We prove that this result holds for each quantum convolutional channel for invertible matrix $G$ with $g_{01}g_{10}\ne0$. This property makes the channel entanglement-breaking, hence anti-degradable. From this, zero private capacity follows, and the extension to mixtures of pure stabilizer states can then be explained by convexity.

\begin{proof}

We first prove that for a pure stabilizer environmental state, the corresponding quantum convolutional channel is entanglement-breaking.


 Let $\sigma$ be a pure stabilizer state, then there exists $\vec{v}\in V^n$ such that \cite{Gross2006}
\begin{align*}
\sigma=\frac{1}{d^n}\sum_{\vec{m}\in M}\omega^{\langle\vec{v},\vec{m}\rangle_s}_dw(\vec{m}),
\end{align*}
where $M$ is a maximal isotropic subspace of $V^n$, that is, $\langle\vec{m}_i,\vec{m}_j\rangle_s=0$ for each $\vec{m}_i,\vec{m}_j\in M$ and $|M|=d^n$.

Denote $A_G(\vec{p},\vec{q})=(NG_{11}\vec{p},g_{00}\vec{q})$ and $B_G(\vec{p},\vec{q})=(-Ng_{10}\vec{p},g_{01}\vec{q})$ for $\vec{p},\vec{q}\in\mathbb{Z}^n_d$, then the Proposition \ref{CMD} implies that
\begin{align*}
 \Xi_{\Lambda_{\sigma}}(\vec{z})=\Xi_{\rho}(A_G\vec{z})\Xi_{\sigma}(B_G\vec{z}),
\end{align*}
for each $\vec{z}\in V^n$. Noting that $\Xi_{\sigma}(\vec{m})\ne0$ only for $\vec{m}\in M$, then $\Xi_{\Lambda_{\sigma}}(\vec{z})=0$ for each $B_G(\vec{z})\notin M$. Therefore, the output of the convolutional channel can be written as
\begin{align*}
\Lambda_{\sigma}(\rho)=d^{-n}\sum_{\vec{m}\in M}\Xi_{\rho}(A_G\vec{m})\chi(B_G\vec{m})w(\vec{m}).
\end{align*}

Secondly, since $g_{01}g_{10}\ne0$ and $N\ne0$, the linear transform $B_G$ is invertible. Then, for each $\vec{z}_1=(\vec{p}_1,\vec{q}_1)$ and $\vec{z}_2=(\vec{p}_2,\vec{q}_2)$, one has
\begin{align*}
\langle B_G\vec{z}_1,B_G\vec{z}_2\rangle_s&=-Ng_{10}\vec{p}_1\cdot(g_{01}\vec{q}_2)+g_{01}\vec{q}_1\cdot(Ng_{10}\vec{p}_2)\\
&=\mu\langle\vec{z}_1,\vec{z}_2\rangle_s,
\end{align*}
where $\mu=-Ng_{10}g_{01}\ne0$.
Denoting $K_G:=B^{-1}_G(M)$, then for each $\vec{k}_i\in K_G$, one has $B_G\vec{k}_i\in M$ and
\begin{align*}
\langle\vec{k}_1,\vec{k}_2\rangle_s=-(Ng_{01}g_{10})^{-1}\langle B_G\vec{k}_1,B_G\vec{k}_2\rangle_s=0.
\end{align*}
Moreover, since $B_G$ is invertible, one has
\begin{align*}
  |K_G|=|M|=d^n.
\end{align*}

Furthermore, Define $\mathcal{A}(K_G):=\textmd{span}\{w(\vec{z}):\vec{z}\in K_G\}$, then one has that the elements of $\mathcal{A}(K_G)$ are commute and
\begin{align*}
    \textmd{Ran}(\Lambda_{\sigma})\subset\mathcal{A}(K_G).
\end{align*}
Fixed a linear isomorphism $\kappa:K_G\rightarrow\mathbb{Z}^n_d$, then for each $\vec{u}\in\mathbb{Z}^n_d$, one can define a character in $K_G$ as
\begin{align*}
\alpha_{\vec{u}}(\vec{z}):=\omega^{\vec{u}\cdot\kappa(\vec{z})},
\end{align*}
and
\begin{align*}
    P_{\vec{u}}:=d^{-n}\sum_{\vec{z}\in K_G}\alpha_{\vec{u}}(\vec{z})w(\vec{z}).
\end{align*}
Then, $\{P_{\vec{u}}\}_{\vec{u}\in K_G}$ is a set of rank-one projection satisfying
\begin{align*}
P_{\vec{u}}P_{\vec{v}}&=d^{-2n}\sum_{\vec{z},\vec{y}\in K_G}\alpha_{\vec{u}}(\vec{z})\alpha_{\vec{v}}(\vec{y})w(\vec{z}+\vec{y})P_{\vec{u}}\\
&=d^{-2n}\sum_{\vec{z},\vec{x}\in K_G}\omega^{\vec{u}\cdot\kappa(\vec{z})+\vec{v}\cdot\kappa(\vec{x}-\vec{z})}w(\vec{x})\\
&=d^{-2n}\sum_{\vec{z}\in\mathbb{Z}^n_d}\omega^{(\vec{u}-\vec{v})\cdot\vec{z}}\sum_{\vec{x}\in K_G}\omega^{\vec{v}\cdot\kappa(\vec{x})}w(\vec{x})\\
&=P_{\vec{v}}\delta_{\vec{u},\vec{v}},
\end{align*}
where the third equality follows from the fact that $\kappa$ is an isomorphism and $\sum_{\vec{z}\in\mathbb{Z}^n_d}\omega^{(\vec{u}-\vec{v})\cdot\vec{z}}=\delta_{\vec{u},\vec{v}}$, and
\begin{align*}
\sum_{\vec{u}\in\mathbb{Z}^n_d}P_{\vec{u}}=I.
\end{align*}
As a result, one has
\begin{align*}
\mathcal{A}(K_G)=\bigoplus_{\vec{u}\in\mathbb{Z}^n_d}\mathbb{C}P_{\vec{u}}.
\end{align*}

At last, since $\textmd{Ran}(\Lambda_{\sigma})\subset\mathcal{A}(K_G)$, then for each $\rho$, it holds that
\begin{align*}
 \Lambda_{\sigma}(\rho)=\sum_{\vec{u}\in\mathbb{Z}^n_d}p_{\vec{u}}(\rho)P_{\vec{u}},
\end{align*}
with
\begin{align*}
p_{\vec{u}}(\rho)=\textmd{tr}(P_{\vec{u}}\Lambda_{\sigma}(\rho))=\textmd{tr}(\Lambda^{\ast}_{\sigma}(P_{\vec{u}})\rho).
\end{align*}
Denote $F_{\vec{u}}:=\Lambda^{\ast}_{\sigma}(P_{\vec{u}})$, since $\Lambda^{\ast}_{\sigma}$ is positive and unital \cite{Spehner2014jmp}, one has
\begin{align*}
  F_{\vec{u}}\ge0,
\end{align*}
and
\begin{align*}
\sum_{\vec{u}\in\mathbb{Z}^n_d}F_{\vec{u}}=\Lambda^{\ast}_{\sigma}(I)=I.
\end{align*}
Therefore, $\{F_{\vec{u}}\}$ forms a POVM and
\begin{align*}
\Lambda_{\sigma}(\rho)=\sum_{\vec{u}\in\mathbb{Z}^n_d}\textmd{tr}(F_{\vec{u}}\rho)P_{\vec{u}}
\end{align*}
is a measure and prepare form, which implies that $\Lambda_{\sigma}$ is entanglement-breaking.

As a result, Theorem 11 in \cite{Cubitt2008jmp} tell us that it is anti-degradable which implies that $P(\Lambda_{\sigma})=P^{(1)}(\Lambda_\sigma)=0$.

On the other hand, since the set of anti-degradable channels is convex, for each mixture of pure stabilizer states $\sigma=\sum_ip_i\ket{\tau}_i\bra{\tau}_i$, the corresponding channel satisfies
\begin{align*}
\Lambda_{\sigma}(\rho)=\Lambda(\rho\otimes\sigma)=\sum_ip_i\Lambda(\rho\otimes\tau_i)=\sum_ip_i\Lambda_{\tau_i}(\rho),
\end{align*}
for any $\rho\in\mathcal{H}_A$, which implies that $\Lambda_{\sigma}$ is still anti-degradable.

In conclusion, for each mixture of pure stabilizer state $\sigma$,
\begin{align*}
P(\Lambda_{\sigma})=P^{(1)}(\Lambda_{\sigma})=0.
\end{align*}
\end{proof}

Theorem \ref{stabilizer-derives-zero-capacity} implies that the stabilizer environmental states ensure that the private capacity and quantum capacity are both zero for any invertible parameter matrix $G$ with $g_{01}g_{10}\ne0$, including a discrete beam splitter. This result also generalizes the main Theorem in \cite{Bu2025prl}.





From Theorem \ref{stabilizer-derives-zero-capacity}, we can also infer that a magic environmental state is necessary to achieve a nonzero private capacity. Therefore, let us consider the case when the environmental state is magic.


\begin{exmp} (Discrete amplifier) Let $G=\begin{pmatrix}
l&-m\\-m&l
\end{pmatrix}$ with $l^2-m^2\equiv1\mod d$, then the discrete amplifier unitary is  $U_{l,m}=\sum_{\vec{i},\vec{j}}\ket{l\vec{i}+m\vec{j},m\vec{i}+l\vec{j}}\bra{\vec{i},\vec{j}}$ and the quantum convolutional channel can be written as
\begin{align*}
\Lambda_{l,m,\sigma}(\rho)=\textmd{tr}_B[U_{l,m}(\rho\otimes\sigma)U^{\dagger}_{l,m}].
\end{align*}

If there exists a magic state environment such that the private information of the corresponding convolutional channel is non-zero, this indicates that magic resources are necessary for the convolutional channel to achieve the capacity of transmitting private information. Actually, choosing a magic state to be
 \begin{align*}
 \sigma=\sqrt{\frac{1}{3}}\ket{0}+\sqrt{\frac{2}{3}}\ket{1},
 \end{align*}
 and let the input ensemble $\{\ket{0}, \ket{-lm^{-1}}\}$ with equal probability, then
\begin{align*}
\Lambda_{l,m,\sigma}(\ket{0})&=\frac{1}{3}\ket{0}\bra{0}+\frac{2}{3}\ket{m}\bra{m},\\
\Lambda_{l,m,\sigma}(\ket{-lm^{-1}})&=\frac{1}{3}\ket{-l^2m^{-1}}\bra{-l^2m^{-1}}\\
&+\frac{2}{3}\ket{-m^{-1}}\bra{-m^{-1}}.
\end{align*}
Hence, the Holevo information of $\Lambda_{l,\sigma}$ is
\begin{align*}
&\chi_{\{\frac{1}{2}, \ket{0}, \ket{-lm^{-1}}\}}(\Lambda_{l,m,\sigma})\\
=&S(\Lambda_{l,m,\sigma}(\frac{1}{2}\ket{0}\bra{0}+\frac{1}{2}\ket{-lm^{-1}}\bra{-lm^{-1}}))\\
&-\frac{1}{2}[S(\Lambda_{l,m,\sigma}(\ket{0}))+S(\Lambda_{l,m,\sigma}(\ket{-lm^{-1}}))]\\
=&H(\frac{1}{6},\frac{1}{6},\frac{1}{3},\frac{1}{3})-\frac{1}{2}[H(\frac{1}{3},\frac{2}{3})+H(\frac{1}{3},\frac{2}{3})],
\end{align*}
where $H(\vec{p}):=-\sum_jp_j\log p_j$ is the Shannon entropy of probability vector $\vec{p}=(p_1,...,p_n)$.

Similarly, the output of $\Lambda^c_{l,\sigma}$ are
\begin{align*}
&\Lambda^c_{l,m,\sigma}(\ket{0})=\frac{1}{3}\ket{0}\bra{0}+\frac{2}{3}\ket{l}\bra{l},\\
&\Lambda^c_{l,m,\sigma}(\ket{-lm^{-1}})=\frac{1}{3}\ket{-l}\bra{-l}+\frac{2}{3}\ket{0}\bra{0}.
\end{align*}
And, the Holevo information of $\Lambda^c_{l,\sigma}$ is
\begin{align*}
&\chi_{\{\frac{1}{2}, \ket{0}, \ket{-lm^{-1}}\}}(\Lambda^c_{l,m,\sigma})\\
=&S(\Lambda^c_{l,m,\sigma}(\frac{1}{2}\ket{0}\bra{0}+\frac{1}{2}\ket{-lm^{-1}}\bra{-lm^{-1}}))\\
&-\frac{1}{2}[S(\Lambda^c_{l,m,\sigma}(\ket{0}))+S(\Lambda_{l,m,\sigma}(\ket{-lm^{-1}}))]\\
=&H(\frac{1}{6},\frac{1}{3},\frac{1}{2})-\frac{1}{2}[H(\frac{1}{3},\frac{2}{3})+H(\frac{1}{3},\frac{2}{3})]
\end{align*}
In fact, the outputs of $\ket{0}$ and $\ket{-lm^{-1}}$ under the complementary channel can be merged, thus resulting in a smaller Shannon entropy. Therefore, for $l^2\neq -m^2 \mod d$, the private capacity of $\Lambda_{l,\sigma}$ satisfies
\begin{align*}
P(\Lambda_{l,m,\sigma})&\ge\chi_{\{\frac{1}{2},\ket{0}, \ket{\frac{-l}{m}}\}}(\Lambda_{l,m,\sigma})-\chi_{\{\frac{1}{2},\ket{0}, \ket{\frac{-l}{m}}\}}(\Lambda^c_{l,m,\sigma})\\
&=H(\frac{1}{6},\frac{1}{6},\frac{1}{3},\frac{1}{3})-H(\frac{1}{6},\frac{1}{3},\frac{1}{2})\\
&=\frac{1}{2}\log3-\frac{1}{3}>\frac{1}{6}.
\end{align*}

\end{exmp}

\begin{exmp}(Discrete beam splitter)\label{discrete-beam-splitter}
Let $d$ an odd prime number and $G=\begin{pmatrix}
s&t\\-t&s
\end{pmatrix}$ with $s^2+t^2=1 \mod d$, then the corresponding unitary is $U_{s,t}=\sum_{\vec{i},\vec{j}}\ket{s\vec{i}+t\vec{j},-t\vec{i}+s\vec{j}}\bra{\vec{i},\vec{j}}$, and the corresponding quantum convolutional channel is
\begin{align*}
\Lambda_{s,t,\sigma}(\rho)=\textmd{tr}_B[U_{s,t}(\rho\otimes\sigma)U^{\dagger}_{s,t}].
\end{align*}
Similarly, by choosing a one-qudit magic state as
\begin{align*}
\sigma=\frac{1}{\sqrt{2}}(\ket{0}+\ket{1}),
\end{align*}
since $Q(\Lambda_{s,t,\sigma})\ge\frac{1}{2}$ \cite{Bu2025prl}, one has
\begin{align*}
P(\Lambda_{s,t,\sigma})\ge Q(\Lambda_{s,t,\sigma})\ge\frac{1}{2}.
\end{align*}
\end{exmp}

\begin{exmp}\label{balanced-convolutional-unitary}
Let $d$ an odd prime number and $G=\begin{pmatrix}
g&h\\-g&h
\end{pmatrix}$ with $g,h\ne0 \mod d$, then the corresponding unitary is $U_{g,h}=\sum_{\vec{i},\vec{j}}\ket{Nh\vec{i}+Ng\vec{j},-Nh\vec{i}+Ng\vec{j}}\bra{\vec{i},\vec{j}}$, and the corresponding channel is defined as
\begin{align*}
 \Lambda_{g,h,\sigma}(\rho)=\textmd{tr}_B[U_{g,h}(\rho\otimes\sigma) U^{\dagger}_{g,h}].
\end{align*}where $N=(2gh)^{-1} \mod d$. 
By choosing a one-qudit magic state as
\begin{align*}
\sigma=\sqrt{\frac{1}{3}}\ket{0}+\sqrt{\frac{2}{3}}\ket{1},
\end{align*}
the output states for the input state $\ket{1}$ are
 \begin{align*}
 &\Lambda_{g,h,\sigma}(\ket{1})=\frac{1}{3}\ket{Nh}\bra{Nh}+\frac{2}{3}\ket{Nh+Ng}\bra{Nh+Ng},\\
 &\Lambda^c_{g,h,\sigma}(\ket{1})=\frac{1}{3}\ket{-Nh}\bra{-Nh}+\frac{2}{3}\ket{N(g-h)}\bra{N(g-h)}.
 \end{align*}
And, the output states for $\ket{1-h^{-1}g}$ are
 \begin{align*}
 \Lambda_{g,h,\sigma}(\ket{1-h^{-1}g})=&\frac{1}{3}\ket{Nh-Ng}\bra{Nh-Ng}+\frac{2}{3}\ket{Nh}\bra{Nh},\\
 \Lambda^c_{g,h,\sigma}(\ket{1-h^{-1}g})=&\frac{1}{3}\ket{N(g-h)}\bra{N(g-h)}\\
 +&\frac{2}{3}\ket{N(2g-h)}\bra{2Ng-Nh}.
 \end{align*}
Therefore, for the input $\ket{1}$ with probability  $\frac{1}{3}$ and input $\ket{1-h^{-1}g}\}$ with probability $\frac{2}{3}$, the private capacity satisfies
\begin{align*}
P(\Lambda_{g,h,\sigma}) &\ge\chi_{\{\frac{1}{3},\ket{1};\frac{2}{3},\ket{1-\frac{g}{h}}\}}(\Lambda_{g,h,\sigma})-\chi_{\{\frac{1}{3},\ket{1};\frac{2}{3},\ket{1-\frac{g}{h}}\}}(\Lambda^c_{g,h,\sigma})\\
&= H(\frac{5}{9},\frac{2}{9},\frac{2}{9})-H(\frac{4}{9},\frac{4}{9},\frac{1}{9})\approx0.043>0.
\end{align*}
\end{exmp}

\begin{rem}
For $d=3$ or $d=5$, and $G=\begin{pmatrix}2&1\\1&1
\end{pmatrix}$,  there exists a magic environmental state such that the convolutional channel has non-zero private capacity. However, for discrete amplifier and discrete beam splitter convolution, the minimal dimension of the quantum system is $7$.

Noting that in \cite{Sun2025pra}, the authors show that for $d=3$ or $d=5$, there also exists a magic environmental state such that the convolutional channel has non-zero quantum capacity.
\end{rem}

Furthermore, for a discrete beam splitter, the corresponding quantum capacity has been proved to be bounded by the amount of magic of $\sigma$, i.e., $\textmd{MRM}(\sigma)$ \cite{Bu2025prl}. Even though the private capacity is an upper bound of quantum capacity, we find that the private capacity of $\Lambda_{\sigma}$ is still upper bounded by $\textmd{MRM}(\sigma)$, up to a constant factor.

\begin{thm}\label{magic-bound-capacity}
(Magic bound private capacity)~For nontrivial $s,t\in\mathbb{Z}_d$ with $s^2+t^2\equiv1 \mod d$ and pure state $\sigma$, it holds that
\begin{align*}
P(\Lambda_{s,t,\sigma})\le 2\textmd{MRM}(\sigma),
\end{align*}
where $\textmd{MRM}(\sigma):=S(\mathcal{M}(\sigma))-S(\sigma)$ is a magic measure.
\end{thm}

\begin{proof}
Noting that the private information of a quantum channel is upper-bounded by the sum of the coherent information of itself and its complementary which is also a discrete beam splitter channel, therefore, the private information of a discrete beam splitter is bounded by $2\textmd{MRM}(\sigma)$ for a pure environmental state. Since the tensor product of discrete beam splitter channel and the environemnt is still pure, then this upper bound also holds for the private capacity.

On one hand, based on the definition of private information and coherent information, it holds that
\begin{align*}
P^{(1)}(\Lambda_{s,t,\sigma})&=\max_{p_i,\rho_i}\{S(\Lambda_{s,t,\sigma}(\sum_ip_i\rho_i))-\sum_ip_iS(\Lambda_{s,t,\sigma}(\rho_i))]\\
&-[S(\sum_ip_i\Lambda^c_{s,t,\sigma}(\rho_i))-\sum_ip_iS(\Lambda^c_{s,t,\sigma}(\rho_i))]\}\\
&\le \max_{\rho}[S(\Lambda_{s,t,\sigma}(\rho))-S(\Lambda^c_{s,t,\sigma}(\rho))]\\
&+\max_{p_i,\rho_i}\sum_ip_i[S(\Lambda^c_{s,t,\sigma}(\rho_i))-S(\Lambda_{s,t,\sigma}(\rho_i))]\\
&\le Q^{(1)}(\Lambda_{s,t,\sigma})+Q^{(1)}(\Lambda^c_{s,t,\sigma}),
\end{align*}
where the first inequality holds because $\Lambda_{s,t,\sigma}$ and $\Lambda^c_{s,t,\sigma}$ are complementary to each other.

On the other hand, since $\sigma$ is a pure state, then the complementary channel of $\Lambda_{s,t,\sigma}$ can be written as (see Proposition \ref{char-function-property})
\begin{align*}
\Lambda^c_{s,t,\sigma}(\rho)=\textmd{tr}_A[U_{s,t}(\rho_A\otimes\sigma_E) U^{\dagger}_{s,t}]=\Lambda_{-t,s,\mathcal{A}(\sigma)}
\end{align*}
then the Theorem 4 in \cite{Bu2025prl} implies that
\begin{align*}
Q(\Lambda^c_{s,t,\sigma})=Q(\Lambda_{-t,s,\mathcal{A}(\sigma)})\le \textmd{MRM}(\mathcal{A}(\sigma))=\textmd{MRM}(\sigma),
\end{align*}
where $\mathcal{A}(\sigma)=A\sigma A^{\dagger}$ with $A=\sum_{\vec{x}}\ket{-\vec{x}}\bra{\vec{x}}$.

As a result, the private information of $\Lambda_{s,\sigma}$ satisfies
\begin{align*}
 P^{(1)}(\Lambda_{s,t,\sigma})\le Q^{(1)}(\Lambda_{s,t,\sigma})+Q^{(1)}(\Lambda^c_{s,t,\sigma})\le2\textmd{MRM}(\sigma).
\end{align*}

Secondly, let us consider the regularization of private information. Noting that $\Lambda^{\otimes N}_{\sigma}$ is still a discrete beam splitter channel with pure environmental state $\sigma^{\otimes N}$, and the corresponding mean state satisfies
\begin{align*}
  \mathcal{M}(\sigma^{\otimes N})= \mathcal{M}(\sigma)^{\otimes N},
\end{align*}
then
\begin{align*}
\textmd{MRM}(\sigma^{\otimes N})=S(\mathcal{M}(\sigma^{\otimes N}))-S(\sigma^{\otimes N})=N\textmd{MRM}(\sigma).
\end{align*}

In conclusion,
\begin{align*}
&P(\Lambda_{s,\sigma})=\lim_{N\rightarrow\infty}\frac{1}{N}P^{(1)}(\Lambda^{\otimes N}_{\sigma})\\
\le&\lim_{N\rightarrow\infty}\frac{1}{N}2N\cdot\textmd{MRM}(\sigma)=2\textmd{MRM}(\sigma).
\end{align*}
\end{proof}

Theorem \ref{magic-bound-capacity} implies that the magic of $\sigma$ should be large enough to get a higher private capacity of $\Lambda_{s,\sigma}$. It is also interesting to find the optimal state $\sigma$ that can achieve the maximal private capacity.

Let $\sigma_k$  be a quantum state generated by a Clifford circuit $U_c$ on $k$ copies of 1 qudit magic state $\ket{\sigma}=\frac{1}{\sqrt{2}}(\ket{0}+\ket{1})$, that is,
\begin{align*}
\sigma_k=U_c(\ket{\sigma}\bra{\sigma}^{\otimes k}\otimes\ket{0}\bra{0}^{\otimes n-k})U^{\dagger}_c,
\end{align*}
then the magic measure is $\textmd{MRM}(\sigma)=k\log d$. And, based on the Theorem 9 in \cite{Bu2025prl}, it holds that $P(\Lambda_{s,t,\sigma})\ge Q(\Lambda_{s,t,\sigma})\ge ck$, where $c$ is a constant. As a result, one has
\begin{align*}
kc\le P(\Lambda_{s,t,\sigma})\le2k\log d.
\end{align*}
This example indicates that $\sigma_k$ is the optimal state satisfying the above equation, up to some constant factor.

At last, we show that for a class of convolutional unitary (Example \ref{balanced-convolutional-unitary}), the private capacity can be zero for some magic environmental state as follows.



We call a quantum state $\sigma$ symmetric if its characteristic functions are real, i.e.,  $\Xi_{\sigma}(\vec{p},\vec{q})=\Xi_{\sigma}(-\vec{p},-\vec{q})$ for each $\vec{p},\vec{q}\in\mathbb{Z}^n_d$. Noting that each probability combination of a pure stabilizer state is symmetric. Denoting that a quantum channel $\mathcal{A}$ as
\begin{align}\label{symmteric-channel}
\mathcal{A}(\rho)=A\rho A^{\dagger},
\end{align}
where $A=\sum_{\vec{j}\in \mathbb{Z}^n_d}\ket{-\vec{j}}\bra{\vec{j}}$. Hence, the fact $\mathcal{A}(w(\vec{p},\vec{q}))=w(-\vec{p},-\vec{q})$ implies that $\Xi_{\mathcal{A}(\rho)}(\vec{p},\vec{q})=\Xi_{\rho}(-\vec{p},-\vec{q})$ for each $\vec{p},\vec{q}\in\mathbb{Z}^N_d.$ Therefore, a quantum state $\sigma$ is symmetric can be also rewritten as $\mathcal{A}(\rho)=\rho $.

\begin{thm}\label{symmetry-limit-capacity} (Symmetry can limit private capacity). Let $d$ be an odd prime number and $\sigma=\sum_ip_iw(\vec{x}_i,\vec{y}_i)\sigma_iw^{\dagger}(\vec{x}_i,\vec{y}_i)$ be an $n$-qudit state with each pure state $\sigma_i$ is symmetric and $w(\vec{x}_i,\vec{y}_i)$s are arbitrary Weyl operators. Then, $\Lambda_{\sigma}$ is anti-degradable for $G=\begin{pmatrix}
g&h\\-g&h
\end{pmatrix}$, and the private capacity
\begin{align*}
P(\Lambda_{g,h,\sigma})=Q(\Lambda_{g,h,\sigma})=0.
\end{align*}
\end{thm}

\begin{proof}

We first show that if the environmental state is a symmetric state under a Clifford transformation, then the corresponding quantum convolutional channel is anti-degradable.

Recall the convolution-multiplication duality of quantum convolution (see Propositions \ref{CMD} and \ref{char-function-property}), one has
\begin{align*}
&\Xi_{\Lambda_{\sigma}(\rho)}(\vec{p},\vec{q})=\Xi_{\rho}(Nh\vec{p},g\vec{q})\Xi_{\sigma}(Ng\vec{p},h\vec{q}),\\
&\Xi_{\tilde{\Lambda}_{\sigma}(\rho)}(\vec{p},\vec{q})=\Xi_{\rho}(-Nh\vec{p},-g\vec{q})\Xi_{\sigma}(Ng\vec{p},h\vec{q}),
\end{align*}
where $N=(2gh)^{-1}$. For each Wely operator $w(\vec{x},\vec{y})$ and $n$-qudit state $\sigma$, denote the unitary channel $D_{(\vec{x},\vec{y})}(\sigma)=w(\vec{x},\vec{y})\sigma w^{\dagger}(\vec{x},\vec{y})$, then
\begin{align*}
&\Xi_{\Lambda_{D_{(\vec{x},\vec{y})}(\sigma)}(\rho)}(\vec{p},\vec{q})=\Xi_{\rho}(Nh\vec{p},g\vec{q})\Xi_{D_{(\vec{x},\vec{y})}(\sigma)}(Ng\vec{p},h\vec{q})\\
=&\omega^{\langle(\vec{x},\vec{y}),(Ng\vec{p},h\vec{q})\rangle_s}_d\Xi_{\rho}(Nh\vec{p},g\vec{q})\Xi_{\sigma}(Ng\vec{p},h\vec{q}).\\
\end{align*}


\begin{widetext}
Then, for any input $\rho$,
\begin{align*}
&D_{(2h\vec{x},2Ng\vec{y})}\circ\mathcal{A}\circ\Lambda_{D_{(\vec{x},\vec{y})}(\sigma)}(\rho)\\
=&\sum_{\vec{p},\vec{q}\in\mathbb{Z}^n_d}\omega^{\langle(\vec{x},\vec{y}),(Ng\vec{p},h\vec{q})\rangle_s}_d\Xi_{\rho}(Nh\vec{p},g\vec{q})\Xi_{\sigma}(Ng\vec{p},h\vec{q})D_{(2h\vec{x},2Ng\vec{y})}[w(-\vec{p},-\vec{q})]\\
=&\sum_{\vec{p},\vec{q}\in\mathbb{Z}^n_d}\omega^{\langle(\vec{x},\vec{y}),(-Ng\vec{p},-h\vec{q})\rangle_s}_d\Xi_{\rho}(Nh\vec{p},g\vec{q})\Xi_{\sigma}(Ng\vec{p},h\vec{q})w(-\vec{p},-\vec{q})\\
=&\sum_{\vec{p},\vec{q}\in\mathbb{Z}^n_d}\omega^{\langle(\vec{x},\vec{y}),(Ng\vec{p},h\vec{q})\rangle_s}_d\Xi_{\rho}(-Nh\vec{p},-g\vec{q})\Xi_{\sigma}(Ng\vec{p},h\vec{q})w(\vec{p},\vec{q})\\
=&\tilde{\Lambda}_{D_{(\vec{x},\vec{y})}(\sigma)}(\rho).
\end{align*}
\end{widetext}
Noting that if $\sigma_i$ is a pure state, then $D_{(\vec{x},\vec{y})}(\sigma)$ is also a pure state and $\tilde{\Lambda}_{D_{(\vec{x},\vec{y})}(\sigma)}$ is the complementary channel of $\Lambda_{D_{(\vec{x},\vec{y})}(\sigma)}$. As a result,  $\Lambda_{D_{(\vec{x},\vec{y})}(\sigma)}$ is both degradable and anti-degradable for any $\vec{x},\vec{y}\in\mathbb{Z}^n_d$.

At last, since the set of anti-degradable channels is convex \cite{Cubitt2008jmp}, then
  $\Lambda_{\sigma}=\sum_ip_i\Lambda_{D_{(\vec{x}_i,\vec{y}_i)}(\sigma_i)} $ is also anti-degradable.



\end{proof}

\begin{exmp}\label{symmetric-environment}
Let $d$ be an odd prime number, then the one-qudit
\begin{align*}
\ket{\sigma}=\frac{1}{\sqrt{3}}(\ket{0}+\ket{1}+\ket{-1\mod d})
\end{align*}
is magic. Moreover, it satisfies $\mathcal{A}(\sigma)=\sigma$. Then, Theorem \ref{symmetry-limit-capacity} implies that $\Lambda_{g,\sigma}$ has zero private capacity.
\end{exmp}

\begin{rem}


 As shown in Theorem \ref{stabilizer-derives-zero-capacity} and Example \ref{symmetric-environment}, the corresponding channels for these two types of environments are anti-degradable, so the private capacity coincides with the quantum capacity.

 In fact, separating private capacity from quantum capacity is inherently challenging, and whether the methods of Horodecki and Smith et al. \cite{Horodecki2008tit,Graeme2008sci} can be integrated into the framework of quantum convolution theory remains an interesting question.


\end{rem}

\section{Conclusion and Future Work}

In this paper, we discuss the relationship between private capacity and magic resource in the recently proposed quantum convolution theory of quantum states \cite{Bu2023pnas, Bu2023discrete}. We show that a magic resource can increase the private capacity of quantum convolutional channels by showing that magic environmental states are necessary but not sufficient to obtain nonzero private capacity for a large class of quantum convolution channels.


These phenomena of private capacity are similar to the case of quantum capacity \cite{Bu2025prl}, hence, a natural question is whether it is possible to separate the private capacity from the quantum capacity within the quantum convolution framework? In particular, does there exist a quantum convolutional channel with zero quantum capacity but positive private capacity?\\





\section{ACKNOWLEDGEMENTS}

We are very grateful to the reviewers for their thorough and critical review, which greatly enhanced the clarity and rigor of our work. This work is supported by the Fundamental Research Funds for the National Natural Science Foundation of China (Grants No.12201555), the Fundamental Research Funds for the Central Universities (Grants
No. 3072025YC2404), the National Natural Science Foundation of Hunan Province (Grants No.2025JJ50050) and Hunan Basic Science Research Center for Mathematical Analysis (2024JC2002).




\appendix

\begin{widetext}

\section{The properties of quantum convolution $\tilde{\boxtimes}$}

Recalling that the key unitary is Clifford, that is,
\begin{align*}
Uw(\vec{p}_1,\vec{q}_1)\otimes w(\vec{p}_2,\vec{q}_2)U^{\dagger}=w(g_{00}\vec{p}_1+g_{01}\vec{p}_2,Ng_{11}\vec{q}_1-Ng_{10}\vec{q}_2)
\otimes w(g_{10}\vec{p}_1+g_{11}\vec{p}_2,-Ng_{01}\vec{q}_1+Ng_{00}\vec{q}_2),
\end{align*}
for any $\vec{p}_i,\vec{q}_i\in \mathbb{Z}^n_d,i=1,2$. Then, it holds that
\begin{align*}
U^{\dagger}w(\vec{p}_1,\vec{q}_1)\otimes w(\vec{p}_2,\vec{q}_2)U
=w(Ng_{11}\vec{p}_1-Ng_{01}\vec{p}_2,g_{00}\vec{q}_1+g_{10}\vec{q}_2)
\otimes w(Ng_{00}\vec{p}_2-Ng_{10}\vec{p}_1,g_{11}\vec{q}_2+g_{01}\vec{q}_1),
\end{align*}

\begin{prop}\label{char-function-property}
For quantum convolution $\tilde{\boxtimes}$, there are several properties as follows.

(1) Convolution-multiplication duality of  $\tilde{\boxtimes}$ is
\begin{align*}
\Xi_{\rho\tilde{\boxtimes}\sigma}(\vec{p},\vec{q})=\Xi_{\rho}(-Ng_{01}\vec{p},g_{10}\vec{q})\Xi_{\sigma}(Ng_{00}\vec{p},g_{11}\vec{q}), \forall \vec{p},\vec{q}\in\mathbb{Z}^n_d.
\end{align*}

(2) For discrete beam splitter $U_{s,t}$ and pure state $\sigma$, it holds that

\begin{align*}
\Lambda^c_{s,\sigma}=\Lambda_{-t,\mathcal{A}(\sigma)},
\end{align*}
where $\mathcal{A}(\rho)=A\rho A^{\dagger}$ (See Equation (\ref{symmteric-channel})).
\end{prop}
\begin{proof}
(1)
\begin{align*}
\Xi_{\rho\tilde{\boxtimes}\sigma}(\vec{p},\vec{q})&=\textmd{tr}[\textmd{tr}_AU(\rho\otimes\sigma)U^{\dagger}w(-\vec{p},-\vec{q})]\\
&=\textmd{tr}[(\rho\otimes\sigma)U^{\dagger}(I_A\otimes w(-\vec{p},-\vec{q}))U]\\
&=\textmd{tr}[(\rho\otimes\sigma)(w(Ng_{01}\vec{p},-g_{10}\vec{q})\otimes w(-Ng_{00}\vec{p},-g_{11}\vec{q}))]\\
&=\Xi_{\rho}(-Ng_{01}\vec{p},g_{10}\vec{q})\Xi_{\sigma}(Ng_{00}\vec{p},g_{11}\vec{q}).
\end{align*}


(2) Noting that for discrete beam splitter, it holds that $\Xi_{\rho\tilde{\boxtimes}\sigma}(\vec{p},\vec{q})=\Xi_{\rho}(-t\vec{p},-t\vec{q})\Xi_{\sigma}(s\vec{p},s\vec{q})$, then for input state $\rho$, one has
\begin{align*}
\Lambda^c_{s,\sigma}(\rho)=\rho\tilde{\boxtimes}\sigma=\sum_{\vec{p},\vec{q}}\Xi_{\rho\tilde{\boxtimes}\sigma}(\vec{p},\vec{q})w(-\vec{p},-\vec{q})=\sum_{\vec{p},\vec{q}}\Xi_{\rho}(-t\vec{p},-t\vec{q})\Xi_{\sigma}(s\vec{p},s\vec{q})w(-\vec{p},-\vec{q})=\Lambda_{-t,s,\mathcal{A}(\sigma)}(\rho),
\end{align*}
where $\Xi_{\mathcal{A}(\rho)}(\vec{p},\vec{q})=\Xi_{\rho}(-\vec{p},-\vec{q})$.
\end{proof}

\end{widetext}

\end{document}